\documentclass{article}

\usepackage{amssymb, amsmath, amsthm, tikz, verbatim, graphicx,lmodern,indentfirst,enumerate,color}

\usepackage[titletoc,page]{appendix}
\usepackage[margin=1.4in]{geometry}

\usepackage{pgf}

\newtheorem{thm}{Theorem}[section]
\newtheorem{lem}[thm]{Lemma}
\newtheorem{prop}[thm]{Proposition}

\theoremstyle{definition}

\theoremstyle{definition}
\newtheorem{defn}[thm]{Definition}

\theoremstyle{remark}
\newtheorem{rem}[thm]{Remark}

\numberwithin{equation}{section}

\makeatletter
\newcommand{\rmnum}[1]{\romannumeral #1}
\newcommand{\Rmnum}[1]{\expandafter\@slowromancap\romannumeral #1@}
\newcommand\restr[2]{{
  \left.\kern-\nulldelimiterspace 
  #1 
  \right|_{#2} 
  }}

\makeatother
\begin{document}

\title{A necessary and sufficient condition for the existence of chaotic dynamics in an overlapping generations model}
\author{Tomohiro Uchiyama\\
Faculty of International Liberal Arts, Soka University,\\ 
1-236 Tangi-machi, Hachioji-shi, Tokyo 192-8577, Japan\\
\texttt{Email:t.uchiyama2170@gmail.com}}
\date{}
\maketitle 

\begin{abstract}
In this paper, we study economic dynamics in a standard overlapping generations model without production. In particular, using numerical methods, we obtain a necessary and sufficient condition for the existence of a topological chaos. This is a new application of a recent result characterising the existence of a topological chaos for a unimodal interval map by Deng, Khan, Mitra (2022). 
\end{abstract}

\section{Introduction}
In this paper, we consider economic dynamics in a standard infinite horizon overlapping generations (OLG) model without production studied in~\cite{Samuelson-consumption-JPE},~\cite{Gale-dynamics-JET},~\cite{Benhabib-Erratic-JEDC} for example. In~\cite[Thm.~1]{Benhabib-Erratic-JEDC}, Benhabib and Day gave a sufficient condition for the existence of a Li-Yorke chaos (that always comes with periodic points of order three) for a fairly general OLG model. In the same paper, to illustrate their general result~\cite[Thm.~1]{Benhabib-Erratic-JEDC}, Benhabib and Day considered two examples of an OLG model~\cite[Sec.~3.4 (\rmnum{1})(\rmnum{3})]{Benhabib-Erratic-JEDC} and showed that some large parameter values in the utility functions are sufficient to generate chaotic dynamics. 

The main results of this paper (Theorems~\ref{main1} and~\ref{main2}) strengthen this result; namely, we give a \emph{necessary and sufficient} condition for the existence of a \emph{topological chaos} (or a \emph{turbulence}) for each of the OLG models in~\cite[Sec.~3.4 (\rmnum{1})(\rmnum{3})]{Benhabib-Erratic-JEDC}. 
We use a recent result characterising the existence of a topological chaos (or a turbulence) for a unimodal interval map in~\cite[Thms.~2 and~3]{Deng-TopChaos-JET}. We must admit that our argument in this paper is not completely rigorous. In the proofs of our main results (Theorems~\ref{main1} and~\ref{main2}), we had to rely on numerical computations by Python since it was too hard for us to obtain algebraic/analytic proofs. In the following, when we use a numerical method we explicitly say so.

Now we clarify what we mean by a turbulence and a chaos. (There are several definitions of a chaos in the literature.) The following definition is taken from~\cite[Chap.~\Rmnum{2}]{Block-book}:
\begin{defn}
Let $g$ be a continuous map of a closed interval $I$ into itself. We call $g$ \emph{turbulent} if there exist three points, $x_1$, $x_2$, and $x_3$ in $I$ such that $g(x_2)=g(x_1)=x_1$ and $g(x_3)=x_2$ with either $x_1<x_3<x_2$ or $x_2<x_3<x_1$. Moreover, we call $g$ \emph{(topologically) chaotic} if some iterate of $g$ is turbulent.  
\end{defn}
It is known that a map $g$ is chaotic (in our sense) if and only if $g$ has a periodic point whose period is not a power of $2$, see~\cite[Chap.~\Rmnum{2}]{Block-book}. This implies that a map $g$ is chaotic if and only if the topological entropy of $g$ is positive, see~\cite[Chap.~\Rmnum{8}]{Block-book}. See~\cite{Ruette-book} for more characterisations of chaos.

Here, we recall the following key result characterising the existence of a topological chaos and a turbulence~\cite[Thms.~2 and~3]{Deng-TopChaos-JET} since our main results are direct applications of this result. Let $\mathfrak{G}$ be the set of continuous maps from a closed interval $[a, b]$ to itself so that an arbitrary element $g\in \mathfrak{G}$ satisfies the following two properties:
\begin{enumerate}
\item{there exists $m\in (a,b)$ with the map $g$ strictly increasing on $[a,m]$ and strictly decreasing on $[m,b]$.}
\item{$g(a)\geq a$, $g(b)<b$, and $g(x)>x$ for all $x\in(a,m]$.}
\end{enumerate}
For $g\in \mathfrak{G}$, let $\Pi:=\{x\in [m,b]\mid g(x)\in [m,b] \textup{ and } g^2(x)=x\}$. Now we are ready to state~\cite[Thms.~2 and~3]{Deng-TopChaos-JET}:
\begin{prop}\label{ChaosThm}
Let $g\in \mathfrak{G}$. The map $g$ has an odd-period cycle if and only if $g^2(m) < m$ and $g^3(m) < \min\{x\in \Pi\}$ and the second iterate $g^2$ is turbulent if and only if $g^2(m) < m$ and $g^3(m) \leq \max\{x\in \Pi\}$. 
\end{prop}

Finally, we state our main results. In this paper, we study two different dynamics generated by the following difference equations:
\begin{alignat}{2}
x_{t+1}&=f(x_t)=rx_te^{-x_t} \textup{ where $r>0$ and $x_t>0$ for any $t$.}\label{system1}\\
x_{t+1}&=h(x_t)=\lambda x_t(\alpha x_t+1)^{-\beta} \textup{ where $\lambda>0$, $\alpha>0$, $\beta>0$, and $x_t>0$ for any $t$.}\label{system2}
\end{alignat}
In Sections~\ref{SectionOLG} and~\ref{SectionProof}, we explain how the OLG models with particular utility functions yield these difference equations. (These are equations considered in~\cite[Sec.~3.4 (\rmnum{1})(\rmnum{3})]{Benhabib-Erratic-JEDC}.) In~\cite[Sec.~3.4 (\rmnum{1})(\rmnum{3})]{Benhabib-Erratic-JEDC}, the authors stated that a sufficient condition for the existence of a Li-Yorke chaos for the first system is $r>14.765$ and that for the second system is $\lambda \geq 50$ and $\beta\geq 5$. We show that both of these results are wrong (or not precise at least). Remember that a Li-Yorke chaos is stronger than a topological chaos. Here are our main results:

\begin{thm}\label{main1}
If $r>e$ in Equation (\ref{system1}), then the map $f$ has the following properties:
\begin{enumerate}
\item{there exists a closed interval $I$ such that $\restr{f}{I}\in \mathfrak{G}$.}
\item{$f$ has an odd period cycle if and only if $r>16.999$.}
\item{The second iterate $(\restr{f}{I})^2$ is turbulent if and only if $r\geq 16.999$.}
\end{enumerate}
\end{thm}
Note that the condition $r>16.999$ in Theorem~\ref{main1} appeared in~\cite[Ex.~4.2]{Marotto-JMAA} as a sufficient condition for Equation (\ref{system1}) to generate a topological chaos although in~\cite[Ex.~4.2]{Marotto-JMAA} a different method (the existence of a \emph{snap-back repeller}) was used to generate a chaos. In Section~\ref{SectionProof}, we show that the assumption $r>e$ in Theorem~\ref{main1} is the minimum assumption to push $\restr{f}{I}$ into $\mathfrak{G}$. We do not know whether we can still generate chaotic dynamics without this condition. (that means without pushing $\restr{f}{I}$ into $\mathfrak{G}$)

\begin{thm}\label{main2}
For the system (\ref{system2}), fix $\alpha=1$ and $\beta=5$ ($\beta=10$ or $15$). If $\lambda>3.052$ ($\lambda>2.868$ or $\lambda>2.815$ respectively), then the map $h$ in Equation (\ref{system2}) has the following properties:
\begin{enumerate}
\item{there exists a closed interval $E$ such that $\restr{h}{E}\in\mathfrak{G}$.}
\item{$h$ has an odd period cycle if and only if $\lambda>85.08$ ($ \lambda>28.11$ or $\lambda>12.45$ respectively).}
\item{The second iterate $(\restr{h}{E})^2$ is turbulent if and only if 
$\lambda\geq 85.08$ ($\lambda\geq28.11$ or $\lambda\geq12.45$ respectively).}
\end{enumerate}
\end{thm}
Note that in Theorem~\ref{main2}, our bound for $\lambda$ to generate a topological chaos when $\beta=5$ is tighter than the bound in~\cite[Sec.~3.4 (\rmnum{3})]{Benhabib-Erratic-JEDC} ($\lambda>85.08$ vs $\lambda>50$). Also, in~\cite[Sec.~3.4 (\rmnum{3})]{Benhabib-Erratic-JEDC} the authors stated the bound without specifying the $\alpha$ value in Equation (\ref{system2}). In general, the bound for $\lambda$ changes if $\alpha$ changes. Possibly, Benhabib and Day might have used a different $\alpha$ value to obtain the bound for $\lambda$. (We used $\alpha=1$.)

Also note that in Theorem~\ref{main2}, we picked some particular $\beta$ values. Our choices of $\beta$ values are ad hoc: we picked some $\beta$ values just to guess a relationship between $\beta$ and $\lambda$. A general pattern is that when $\beta$ is small, $\lambda$ needs to be large (and vice versa). Ideally, we should obtain a general (algebraic) relation between $\beta$ and $\lambda$, but that turned out to be too difficult. So we picked particular $\beta$ values and relied on some numerical methods to obtain the corresponding $\lambda$ values.  

It is well known that these two systems (\ref{system1}) and (\ref{system2}) are widely studied in economics, in mathematical biology, or in dynamical systems in general, see~\cite{Hassel-Patterns-Animal},~\cite{May-Bifurcation-Naturalist} for example. In dynamical systems, these difference functions are 
often studied together since the system (\ref{system1}) is a limiting case of the system (\ref{system2}) (Take the limits $\beta\rightarrow \infty$, $\alpha\rightarrow 0$ while fixing $\alpha\beta=1$). That is why we study these two systems together in the this paper.

\section{OLG model}\label{SectionOLG}
We quickly review how an OLG model generates economic dynamics following~\cite[Sec.~2]{Benhabib-Erratic-JEDC} and~\cite[Sec.~2]{Gale-dynamics-JET}. We consider a pure consumption-loan model between a population of overlapping generations. To simplify the argument, we assume the following: 1.~ The population does not grow. 2.~Each individual lives for two periods. 3.~Each individual has the same utility function $U(c_0(t),c_1(t+1))$ where $c_0(t)$ and $c_1(t+1)$ are the consumption levels of the individual when the individual is young and when the individual is old respectively. 4.~Each individual receives the same endowment $w_0$ when the individual is young and $w_1$ when the individual is old. We write $\rho_t$ for the interest rate at time $t$. Then each individual faces the budget constraint, that is 
\begin{equation}\label{budget}
c_1(t+1)=w_1+\rho_t[w_0-c_0(t)], \; c_0(t)\geq 0, \; c_1(t+1)\geq 0. 
\end{equation}
For each period $t$, the following market clearing condition needs to be satisfied:
\begin{equation}\label{market}
w_0-c_0(t)+w_1-c_1(t)=0.
\end{equation}

We make further assumptions to generate a well-defined dynamics. First, we assume that the utility function $U(c_0(t),c_1(t+1))$ is strictly concave, twice differentiable, increasing in its arguments and separable. Second, we assume that for each $t$, $c_0(t)>0$ and  $c_1(t)>0$. Under these assumptions, each individual tries to maximise the utility under the budget constraint, then we obtain the first order condition: 
\begin{equation}\label{firstOrder}
\rho_t=\frac{U_0(c_0(t),c_1(t+1))}{U_1(c_0(t),c_1(t+1))} \textup{ where } U_0 \textup{ and } U_1 \textup{ are partial derivatives of $U$.}
\end{equation}
Substitute Equation (\ref{firstOrder}) into the budget constraint (\ref{budget}), we obtain
\begin{equation}\label{Gale}
\frac{U_0(c_0(t),c_1(t+1))}{U_1(c_0(t),c_1(t+1))} = \frac{w_1-c_1(t+1)}{c_0(t)-w_0}.
\end{equation}

Now we make the final important assumption, that is,  $c_0(t)>w_0$ for each $t$. This means that people consume more than their income when they are young and less when they are old. Gale called this the \emph{classical case} in~\cite{Gale-dynamics-JET}. In the classical case, we can solve Equation (\ref{Gale}) for $c_1(t+1)$ uniquely by~\cite[Thm.~5]{Gale-dynamics-JET}. Then we can write
\begin{equation}\label{CMRS}
V(c_0(t),w_0,w_1):=\frac{U_0(c_0(t),c_1(t+1))}{U_1(c_0(t),c_1(t+1))}
\end{equation}
Following~\cite[Sec.~2]{Benhabib-Erratic-JEDC}, we call $V(c_0(t),w_0,w_1)$ \emph{a constrained marginal rate of substitution} (CMRS). Finally, combining Equations (\ref{CMRS}) and (\ref{market}), we obtain the economic dynamics we study in this paper (a difference equation in terms of $c_0$):
\begin{equation}\label{dynamics}
c_0(t+1)=w_0 + V(c_0(t),w_0,w_1)(c_0(t)-w_0).
\end{equation}

In~\cite[Sec.~3]{Benhabib-Erratic-JEDC}, roughly speaking, Benhabib and Day showed that if a CMRS varies sufficiently when $c_0$ varies, then Equation (\ref{dynamics}) possesses a chaotic behavior, see~\cite[Sec.~3.2 and Thm.1]{Benhabib-Erratic-JEDC} for the precise statement. 

Next, we derive Equations (\ref{system1}) and (\ref{system2}) from Equation (\ref{dynamics}) using two particular utility functions. We keep all the assumptions on $c_0(t)$, $c_1(t+1)$, and $w_0$ above. Following~\cite[Sec.~3.4 (\rmnum{1})]{Benhabib-Erratic-JEDC}, we set 
\begin{equation*}
U(c_0(t),c_1(t+1))= A - e^{a\left(1-(c_0(t)-w_0)/a\right)}+c_1(t+1) \textup{ where $A>0$ and $a>0$}.  
\end{equation*}
Then using Equation (\ref{dynamics}), we obtain the difference equation 
\begin{equation*}
c_0(t+1)-w_0=e^{a\left(1-(c_0(t)-w_0)/a\right)}(c_0(t)-w_0). 
\end{equation*}
By the changes of variables (setting $c_0(t)-w_0=x_t$), the dynamics expressed by the last formula is (topologically) equivalent to
\begin{equation*}
x_{t+1}=f(x_t)=e^{a\left(1-(x_t/a)\right)}x_t=r x_t e^{-x_t} \textup{ where $r>0$.}
\end{equation*}
This is our Equation (\ref{system1}). Note that $x_t>0$ for any $t$ since we consider the classical case. 

Finally, we consider Equation (\ref{system2}). Following~\cite[Sec.~3.4 (\rmnum{3})]{Benhabib-Erratic-JEDC}, we set 
\begin{equation*}
U(c_0(t),c_1(t+1))= \frac{\lambda(c_0(t)+b)^{1-\beta}}{1-\beta}+c_1(t)\textup{ where $\beta\geq 0, \beta\neq 1, \lambda>0, b\geq 0$}.  
\end{equation*}
Using Equation (\ref{dynamics}), we obtain the difference equation 
\begin{equation*}
c_0(t+1)-w_0=\lambda\frac{c_0(t)-w_0}{((c_0(t)-w_0)+b+w_0)^\beta}. 
\end{equation*}
By the changes of variables (setting $c_0(t)-w_0=x_t$) and setting $\frac{1}{b+w_0}=\alpha$, the dynamics expressed by the last formula is (topologically) equivalent to
\begin{equation*}
x_{t+1}=\lambda x_t (\alpha x_t + 1)^{-\beta}.
\end{equation*}
This is our Equation (\ref{system2}). In this paper, we normalise $\alpha=1$ to simplify the argument.

\section{Proofs of the main results}\label{SectionProof}
\subsection{Proof of Theorem~\ref{main1}}
First, we have $f'(x) = r e^{-x} - r x e^{-x} = r e^{-x}(1-x)$. Since $r>0$ and $e^{-x}>0$, we have $f'(x)>0$ if $0<x<1$, $f'(x)<0$ if $x>1$, and $f'(x)=0$ if $x=1$. So clearly $f$ is unimodal and has the (global) maximum at $x=1$. Using the notation in Introduction, we write $m=1$. Since we want to use Proposition~\ref{ChaosThm}, we need to push $f$ (or some restriction of it) into $\mathfrak{G}$. In particular, we need $f(x)>x$ for all $x\in (a,m]$ for some $a>0$. So, we need $f(m)>m$ at least, that is, $r>e$. We assume $r>e$ in the rest of the paper. 

Second, to push $f$ into $\mathfrak{G}$, we need to restrict the map $f$ to some closed interval $I=[a,b]$ so that $\restr{f}{I}$ is a map from $I$ to itself. In the following, we show that this can be done by setting $a=f^2(1)$ and $b=f(1)$. It is clear that $\restr{f}{I}(x)\leq b$ for any $x\in I$. So, we just need to show that $a\leq \restr{f}{I}(x)$ for any $x\in I$. Since the map $\restr{f}{I}$ is unimodal, the minimum is attained at $x=a$ or at $x=b$. So we have to show that $a\leq f(a)$ and $a\leq f(b)$. The second inequality is trivial, so we focus on $a\leq f(a)$, that is, $f^2(1) \leq f^3(1)$. By a direct calculation, we see that the last inequality is equivalent to $\ln{r}-r^2 e^{\frac{-r}{e}-1}>0$. Now, a numerical calculation by Python gives $r>e$ (precise to ten decimal places). 

To push $\restr{f}{I}$ into $\mathfrak{G}$, we also need: 1.~$f(b)<b$, 2.~$f(x)>x$ for all $x\in (a,m]$. A direct calculation shows that $f(b)<b$ is equivalent to $\frac{1}{e}-\frac{\ln{r}}{r}>0$. An easy calculus shows that $\frac{1}{e}-\frac{\ln{r}}{r}> 0$ for any $r>0$ except $r=e$ and it attains its minimum value $0$ at $r=e$. We are left to show $f(x)>x$ for all $x\in (a,m]$. By an easy calculation, we see that $f(x)>x$ is equivalent to $r e^{-x}>1$. Since we want to show that the last inequality holds for any $x\in (a,m]$, it is enough to show that it holds for $x=m$. (since $r e^{-m}<r e^{-x}$) Now we have to show $r e^{-m} = \frac{r}{e} >1$, this is certainly true. We have shown the following:
\begin{lem}\label{FirstLemma}
If $r>e$, then $\restr{f}{I}\in \mathfrak{G}$. 
\end{lem}
Now we show that 
\begin{lem}\label{SecondLemma}
If $r>9.549$, then the set $\Pi=\{x\in [m,b]\mid f(x)\in [m,b] \textup{ and } f^2(x)=x\}$ is a singleton, namely, $\Pi=\{\ln{r}\}$ (the unique fixed point for the map $f$). 
\end{lem}
\begin{proof}
First, solving $f(x)=x$, we obtain $x=\ln{r}$ (the unique fixed point of $f$). We write $z=\ln{r}$. (Note that we have assumed $r>e$, so $z>0$ as it should be.) It is clear that $z\in \Pi$. Now an easy calculation shows that $f^2(x)=x$ is equivalent to $x(1+r e^{-x})=2\ln{r}$. Plotting the $(x,r)$ satisfying the last equation, we obtain Figure~\ref{fig1}. In Figure~\ref{fig1}, the blue curve gives the fixed point of $f$ for each $r$, and the red curve gives the period two points for each $r$. A numerical calculation gives that if $r>9.549$, then the smaller of the period two points is not in $[1, b)=[m,b)$.  
\begin{figure}[h]
	\begin{center}
    	\scalebox{.4}{\input{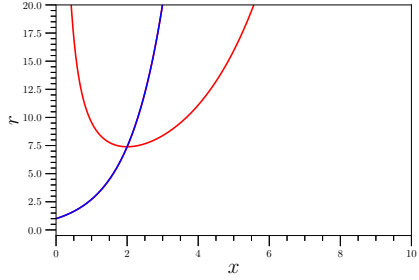}}
	\end{center}
    \caption{$x(1+r e^{-x})=2\ln{r}$}\label{fig1}
\end{figure}
\end{proof}

In view of Proposition~\ref{ChaosThm}, what we need to do now is to translate two conditions, $f^2(m)<m$ and $f^3 < \min\{x\in \Pi\}=z$ in terms of $r$. First, we show that
\begin{lem}\label{ThirdLemma}
$f^2(m)<m$ if and only if $9.549<r$.
\end{lem}
\begin{proof}
We compute that $f^2(m)<m$ is equivalent to $r^2 e^{-\frac{r}{e}-1}<1$. Using Python, the last inequality gives $r<e$ or $9.549<r$. Since we assumed $r>e$, we get the desired result. 
\end{proof}
Next, we show that
\begin{lem}\label{ForthLemma}
Let $r>9.549$. Then $f^3(m) < \min\{x\in \Pi\}=z$ if and only if $r>16.999$. 
\end{lem}
\begin{proof}
A direct calculation gives $f^3(m)-z<0$ is equivalent to $r^3 e^{-\frac{r}{e}-1-r^2 e^{-\frac{r}{e}-1}}-\ln{r}<0$. Now using Python, we obtain $r>16.999$. 
\end{proof}
The same argument yields
\begin{lem}\label{FifthLemma}
Let $r>9.549$. Then $f^3(m) \leq \max\{x\in \Pi\}=z$ if and only if $r\geq16.999$. 
\end{lem}
Finally, combining Proposition~\ref{ChaosThm}, Lemmas~\ref{FirstLemma},~\ref{SecondLemma},~\ref{ThirdLemma},~\ref{ForthLemma}, and~\ref{FifthLemma}, we have proven Theorem~\ref{main1}.

\subsection{Proof of Theorem~\ref{main2}}
Our argument in this proof is similar to that in the proof of Theorem~\ref{main1}, so just give a sketch. First, we have $h'(x)=\lambda \beta(x+1)^{-\beta-1}\left(x+1 -\beta x\right)$, so $h'(x)>0$ if $0<x<\frac{1}{\beta-1}$, $h'(x)<0$ if $\frac{1}{\beta-1}<x$, and $h'(x)=0$ if $x=\frac{1}{\beta-1}$. Thus $h$ is unimodal and takes its maximum at $x=\frac{1}{\beta-1}$. We set $m=\frac{1}{\beta-1}$. Since we want $h$ to be unimodal on $x>0$, we need $m>0$, that is, $\beta>1$. 

Next, we compute the following (since we need these below):
\begin{alignat*}{2}
h(m)&=\frac{\lambda\left(1+\frac{1}{\beta-1}\right)^{-\beta}}{\beta-1},\\
h^2(m)&=\frac{\lambda^2\left(1+\frac{1}{\beta-1}\right)^{-\beta}\left(\frac{\lambda\left(1+\frac{1}{\beta-1}\right)^{-\beta}}{\beta-1}+1\right)^{-\beta}}{\beta-1}, \\
\end{alignat*}
Since we want to push $h$ (or some restriction of it, say, to $[a,b]$) into $\mathfrak{G}$, we need $h(x)>x$ for all $x\in (a,m]$. Therefore we need $h(m)>m$ at least. Solving this, we obtain 
\begin{equation}
\lambda>\left(\frac{1}{\beta-1}+1\right)^\beta.\label{condition1} 
\end{equation}
In view of Proposition~\ref{ChaosThm}, we need $h^2(m)<m$ to generate a chaos. Solving this, we obtain
\begin{equation}
\lambda^2\left(\frac{\beta-1}{\beta}\right)^\beta\left(\frac{\lambda}{\beta-1}\left(\frac{\beta-1}{\beta}\right)^\beta+1\right)^{-\beta}<1.\label{condition2}
\end{equation}
Using python, we plot $(\beta,\lambda)$ satisfying Equations (\ref{condition1}) and (\ref{condition2}) (Figure~\ref{fig2}). 
\begin{figure}[h!]
	\begin{center}
    	\scalebox{.4}{\input{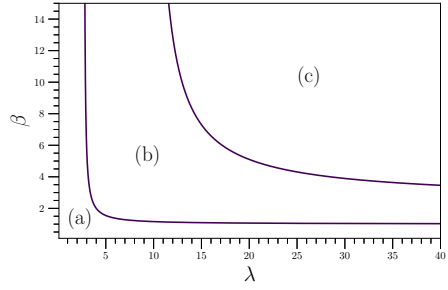}}
	\end{center}
    \caption{$h(m)>m: (b)+(c)$ and $h^2(m)<m: (a)+(c)$}\label{fig2}
\end{figure}
In Figure~\ref{fig2}, regions (b) and (c) correspond to $h(m)>m$, and regions (a) and (c) correspond to $h^2(m)<m$. So, to obtain a chaotic behaviour (using Proposition~\ref{ChaosThm}), we need a $(\lambda, \beta)$ pair in region (c).  

Now, we set $E=[a,b]=[h^2(m), h(m)]$, and consider $\restr{h}{E}$. To make $\restr{h}{E}$ a function from $E$ to itself, by the same argument as in the last example, we need $h^2(m)\leq h^3(m)$. Here, we compute 
\begin{alignat*}{2}
h^3(m)&=\frac{\lambda^3\left(1+\frac{1}{\beta-1}\right)^{-\beta}\left(\frac{\lambda\left(1+\frac{1}{\beta-1}\right)^{-\beta}}{\beta-1}+1\right)^{-\beta}\left(\frac{\lambda^2\left(1+\frac{1}{\beta-1}\right)^{-\beta}\left(\frac{\lambda\left(1+\frac{1}{\beta-1}\right)^{-\beta}}{\beta-1}+1\right)^{-\beta}}{\beta-1}+1\right)^{-\beta}}{\beta-1}.
\end{alignat*}
It is too hard (and messy) to find an algebraic relation between $\lambda$ and $\beta$ from $h^2(m)\leq h^3(m)$ (since $h(m)$ and $h^3(m)$ are too complicated algebraically), so we pick particular $\beta$ values (in region (c) in Figure~\ref{fig2}) and find the condition for $\lambda$ from $h^2(m)\leq h^3(m)$ using numerical computations. We obtain  
\begin{enumerate}
\item{if $\beta=5$, then $h(m)=0.08192\lambda$, $h^2(m)=\frac{0.08192\lambda^2}{(0.08192\lambda+1)^5}$, $h^3(m)=\frac{0.08192\lambda^3}{(0.08192\lambda+1)^5\left(\frac{0.08192\lambda^2}{(0.08192\lambda+1)^5}+1\right)^5}$.}
\item{if $\beta=10$, then $h(m)=0.03874\lambda$, $h^2(m)=\frac{0.03874\lambda^2}{(0.03874\lambda+1)^{10}}$,  $h^3(m)=\frac{0.03874\lambda^3}{(0.03874\lambda+1)^{10}\left(\frac{0.03874\lambda^2}{(0.03874\lambda+1)^{10}}+1\right)^{10}}$.}
\item{if $\beta=15$, then $h(m)=0.02537\lambda, h^2(m)=\frac{0.02537\lambda^2}{(0.02537\lambda+1)^{15}}$,  $h^3(m)=\frac{0.02537\lambda^3}{(0.02537\lambda+1)^{15}\left(\frac{0.02537\lambda^2}{(0.02537\lambda+1)^{15}}+1\right)^{15}}$.}
\end{enumerate}
Now a numerical calculation shows that the relation $h^2(m)\leq h^3(m)$ is equivalent to $\lambda\geq 3.052$ if $\beta=5$, $\lambda\geq 2.868$ if $\beta=10$, and $\lambda\geq 2.815 $ if $\beta=15$.  

Next, to push $\restr{h}{E}$ into $\mathfrak{G}$, we need: 1.~$h(b)<b$, 2.~$h(x)>x$ for all $x\in (a,m]$. By a numerical calculation, we see that $h(b)<b$ (that is $h^2(m)<h(m)$) holds for any $\lambda>0$ and for $\beta=5$ (or $\beta=10$ or $\beta=15$).
This proves the part 1. Now for part 2, we have that $h(x)>x$ is equivalent to $\lambda(x+1)^{-\beta}-1>0$. Since we need to show that $h(x)>x$ for any $x\in (a,m]$, it is enough to show that $\lambda(m+1)^{-\beta}-1=\lambda(\frac{1}{\beta-1}+1)^{-\beta}-1>0$. 
Solving the last inequality numerically, we have $\lambda>3.052$ if $\beta=5$, $\lambda>2.868$ if $\beta=10$, and $\lambda>2.815$ if $\beta=15$. So far we have proven:
\begin{lem}\label{firstLem2}
Let $\beta=5$ ($\beta=10$ or $\beta=15$). If $\lambda>3.052$ ($\lambda>2.868$ or $\lambda>2.815$ respectively), then $\restr{h}{E}\in \mathfrak{G}$. 
\end{lem}
Note that from Figure~\ref{fig2}, we see that $\lambda$ values in Lemma~\ref{firstLem2} are too small to generate chaos. Now we consider two conditions in Proposition~\ref{ChaosThm}, namely, 1.~$h^2(m)<m$ and 2.~$h^3(m)<\min\{x\in \Pi \}$. We consider condition~1 first. Using Equation~\ref{condition2} and Python, we obtain
\begin{lem}\label{secondLem2}
Let $\beta=5$ ($\beta=10$ or $\beta=15$). Then $h^2(m)<m$ if and only if $\lambda>20.45$ ($\lambda>12.97$ or $\lambda>11.57$ respectively).
\end{lem}
Next we look at condition~2 and prove:
\begin{lem}\label{thirdLem2}
Let $\beta=5$ ($\beta=10$ or $\beta=15$). If $\lambda>20.45$ ($\lambda>12.97$ or $\lambda>11.57$ respectively), then the set $\Pi$ is a singleton, namely, $\Pi=\{\lambda^{\frac{1}{\beta}}-1\}$ (the unique fixed point of $h$).
\end{lem}
\begin{proof}
First, solving $h(x)=x$, we obtain the unique fixed point $x=\lambda^{\frac{1}{\beta}}-1$ of $h$. We write $z$ for this fixed point. In the following, we consider each $\beta$ value separately. We consider $\beta=5$ first. We see that $f^2(x)=x$ is equivalent to $\lambda^{\frac{2}{5}}=\lambda x (x+1)^{-4}+x+1$. Plotting $(x,\lambda$ satisfying this equation, we obtain Figure~\ref{fig3} where the blue curve represents the set of fixed points of $h$, and the red curve represents the set of period two points of $h$. 
\begin{figure}[h!]
	\begin{center}
    	\scalebox{.4}{\input{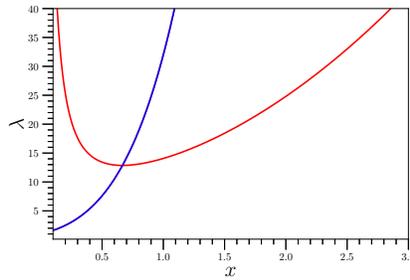}}
	\end{center}
    \caption{$f^2(x)=x$ for $\beta=5$}\label{fig3}
\end{figure}
From Figure~\ref{fig3}), we see that if $\lambda$ is large enough, then the smaller of the period two points of $h$ is less than $m=0.25$. By a numerical computation, we checked that this happens if $\lambda>20.45$. Since $\Pi$ clearly contains $z$, we have shown that if $\lambda>20.45$, then $\Pi=\{z\}$.   

The situation for $\beta=10$ or $\beta=15$ is similar. We just record that if $\beta=10$, then $f^2(x)=x$ is equivalent to $\lambda^{\frac{1}{5}}=\lambda x (x+1)^{-9}+x+1$ and if $\beta=10$ then that is equivalent to $\lambda^{\frac{2}{15}}=\lambda x (x+1)^{-14}+1$. By the same argument as that for the $\beta=5$ case, a numerical method yields that if $\beta=10$ and $\lambda>12.97$ (or $\beta=15$ and $\lambda>11.57$) then the set $\Pi$ is a singleton.  
\end{proof}

\begin{rem}
We do not know why the $\lambda$ values in Lemma~\ref{secondLem2} agree with those in Lemma~\ref{thirdLem2}. We have observed the same phenomenon in the last example ($r$ values in Lemmas~\ref{SecondLemma} and ~\ref{ThirdLemma}). Again, we do not know the reason for this. 
\end{rem}

Finally, we show that
\begin{lem}\label{forthLem2}
Let $\beta=5$ ($\beta=10$ or $\beta=15$) and let $\lambda>20.45$ ($\lambda>12.97$ or $\lambda>11.57$ respectively). Then $h^3(m)<\min\{x\in\Pi\}$ if and only if $\lambda>85.08$ ($\lambda>28.11$ or $\lambda>12.45$ respectively).
\end{lem}
\begin{proof}
We give argument for $\beta=5$. Other cases are similar. Since $\lambda>20.45$, we have $\Pi=\{z\}$. So, $h^3(m)<\min\{x\in \Pi\}$ is equivalent to $h^3(m)<z=\lambda^{\frac{1}{5}}$. Now, using the expression for $h^3(m)$ above, we obtain $\lambda>85.1$ numerically.  
\end{proof}
It is clear that a similar argument yields
\begin{lem}\label{fifthLem2}
Let $\beta=5$ ($\beta=10$ or $\beta=15$) and let $\lambda>20.45$ ($\lambda>12.97$ or $\lambda>11.57$ respectively). Then $h^3(m)\leq\max\{x\in\Pi\}$ if and only if $\lambda\geq85.08$ ($\lambda\geq28.11$ or $\lambda\geq12.45$ respectively).
\end{lem}

Combining Proposition~\ref{ChaosThm} and  Lemmas~\ref{firstLem2},~\ref{secondLem2},~\ref{thirdLem2},~\ref{forthLem2},~\ref{fifthLem2}, we obtain Theorem~\ref{main2}.

\bibliography{econbib}

\begin{thebibliography}{1}

\bibitem{Benhabib-Erratic-JEDC}
J.~Benhabib and R.H. Day.
\newblock A characterization of erratic dynamics in the overlapping generations
  model.
\newblock {\em J. Econ. Dyn. Control}, 4:37--55, 1982.

\bibitem{Block-book}
L.S. Block and W.A. Coppel.
\newblock {\em {Dynamics in One Dimension}}.
\newblock Springer, Berlin, 1992.

\bibitem{Deng-TopChaos-JET}
L.~Deng, M.A. Khan, and T.~Mitra.
\newblock Continuous unimodal maps in economic dynamics: On easily verifiable
  conditions for topological chaos.
\newblock {\em J. Econ. Theory}, 201, 2022.
\newblock Article 105446.

\bibitem{Gale-dynamics-JET}
D.~Gale.
\newblock Pure exchange equilibrium of dynamic economic model.
\newblock {\em J. Econ. Theory}, 6:12--36, 1973.

\bibitem{Hassel-Patterns-Animal}
M.~Hassel, J.~Lawton, and R.M. May.
\newblock Patterns of dynamical behavior in single-species populations.
\newblock {\em J. Anim. Ecol}, 45:471--486, 1976.

\bibitem{Marotto-JMAA}
F.R. Marotto.
\newblock Snap-back repellers imply chaos in {$\mathbb{R}^n$}.
\newblock {\em J. Math. Anal. Appl}, 63:199--223, 1978.

\bibitem{May-Bifurcation-Naturalist}
R.M. May and G.F. Oster.
\newblock Bifurcations and dynamic complexity in simple ecological models.
\newblock {\em Am. Nat}, 110:573--594, 1976.

\bibitem{Ruette-book}
S.~Ruette.
\newblock {\em {Chaos on the Interval}}.
\newblock American Mathematical Society, Providence, 2017.

\bibitem{Samuelson-consumption-JPE}
P.~Samuelson.
\newblock An exact consumption-loan model of interest with or without the
  social contrivance of money.
\newblock {\em J. Political Econ.}, 66:467--482, 1958.

\end{thebibliography}

\end{document}